\documentclass[11pt]{article}
\usepackage[dvipsnames,svgnames,x11names]{xcolor}
\usepackage[utf8]{inputenc}
\usepackage{authblk}

\usepackage{fullpage}
\usepackage{tabularx}
\usepackage{makecell}
\usepackage{amsthm}
\usepackage{amsmath}
\usepackage{amssymb}
\usepackage{amsfonts}
\usepackage{mathtools}
\usepackage{enumitem}
\usepackage[pagebackref]{hyperref}
\usepackage[sort,numbers]{natbib}
\usepackage{xspace}
\usepackage{soul}
\usepackage{color}
\usepackage{tikz}
\usepackage[capitalize]{cleveref}
\usepackage{todonotes}
\newtheorem{theorem}{Theorem}
\newtheorem{lemma}[theorem]{Lemma}

\newtheorem{corollary}[theorem]{Corollary}
\newtheorem{proposition}[theorem]{Proposition}

\theoremstyle{definition}
\newtheorem{definition}[theorem]{Definition}

\newcommand{\problemdef}[3]{
	\begin{center}\fbox{
	\begin{minipage}{0.95\textwidth}
		\noindent
		\underline{#1}
		\vspace{5pt}\\
		\setlength{\tabcolsep}{3pt}
		\begin{tabularx}{\textwidth}{@{}lX@{}}
			\textbf{Input:}     & #2 \\
			\textbf{Question:}  & #3
		\end{tabularx}
	\end{minipage}}
	\end{center}
}
\newcommand{\probnamelong}{\textsc{Constructive Bribery for Challenge the Champ Tournaments}\xspace}
\newcommand{\probnamecup}{\textsc{Constructive Bribery for Cup Tournaments}\xspace}
\newcommand{\probname}{CBCCT\xspace}
\newcommand{\PK}{\textsc{Product Knapsack}\xspace}
\newcommand{\MPK}{\textsc{Multicolored Product Knapsack}\xspace}

\title{Parameterized Analysis of Bribery in Challenge~the~Champ~Tournaments}

\author{Juhi~Chaudhary\thanks{Supported by the ERC, grant number 101039913.}}
\author{Hendrik~Molter\thanks{Supported by the ISF, grant number~1456/18, and the ERC, grant number 949707.}}
\author{Meirav~Zehavi$^*$}

\affil{\small Department of Computer Science, Ben-Gurion~University~of~the~Negev, 
Beer-Sheva, 
Israel\\ \texttt{juhic@post.bgu.ac.il, molterh@post.bgu.ac.il, meiravze@bgu.ac.il}}

\date{}

\begin{document}

\maketitle

\begin{abstract} 
\emph{Challenge the champ tournaments} are one of the simplest forms of competition, where a (initially selected) champ is repeatedly challenged by other players. If a player beats the champ, then that player is considered the new (current) champ. Each player in the competition challenges the current champ once in a fixed order. The champ of the last round is considered the winner of the tournament. We investigate a setting where players can be bribed to lower their winning probability against the initial champ. The goal is to maximize the probability of the initial champ winning the tournament by bribing the other players, while not exceeding a given budget for the bribes. Mattei et al.~[Journal of Applied Logic, 2015] showed that the problem can be solved in pseudo-polynomial time, and that it is in XP when parameterized by the number of players.

We show that the problem is weakly NP-hard and W[1]-hard when parameterized by the number of players. On the algorithmic side, we show that the problem is fixed-parameter tractable when parameterized either by the number of different bribe values or the number of different probability values. To this end, we establish several results that are of independent interest. In particular, we show that the product knapsack problem is W[1]-hard when parameterized by the number of items in the knapsack, and that constructive bribery for cup tournaments is W[1]-hard when parameterized by the number of players. Furthermore, we present a novel way of designing mixed integer linear programs, ensuring optimal solutions where \emph{all} variables are integers.

\bigskip

\end{abstract}
\section{Introduction}

Sports tournaments are ubiquitous at global events such as World Cups and the Olympics, national events such as sports leagues, and local events such as school competitions. 
While entertaining, these sports tournaments aim to impartially identify the most talented player, the {\em champ}, according to specific criteria. Unfortunately, the crucial requirement of ensuring fairness in this process is a highly complicated challenge. On top of the fact that every player aspires to become the champ, the ongoing monetization of sports---through advertising and lucrative brand deals awarded to winners--- intensifies the competition. Accordingly, the performance of various forms of manipulation in tournaments, such as bribery, constitutes a significant body of research in social choice theory and related disciplines.
These works concern, in particular, round-robin tournaments (see, e.g., \cite{baumeister2021complexity,krumer2023strategic,rasmussen2008round}), cup tournaments (see, e.g., \cite{suksompong2021tournaments,williams_moulin_2016,gupta2018winning,russell2009manipulating,vu2009complexity}), and challenge the champ tournaments~\cite{mattei2015complexity}.

The literature focuses on several prominent ways to manipulate a tournament. The (arguably) most natural one is to offer incentives such as bribes to specific players (individuals or part of a team), team coaches, or judges, persuading them to lose (or, in the case of judges, flip the outcome) of a match deliberately~\cite{russell2009manipulating}.  
We focus on the standard concept of budget-constrained bribery in tournaments~\cite{gupta2018winning,russell2009manipulating,vu2009complexity}, and on challenge the champ tournaments (as well as, to some extent, cup tournaments).

\paragraph{Our Setting.} We study the computational problem of constructive (budget-constrained) bribery in challenge the champ tournaments in the (standard) probabilistic setting, termed \probnamelong (\probname). The study of the complexity of this problem was initiated by \citet{mattei2015complexity}. Challenge the champ tournaments consist of a set of $n+1$ players, $\{e_{1},\ldots,e_{n},e^{*}\}$, where $e^{*}$ is the initial champ. The (initially selected) champ~$e^*$ is repeatedly challenged by the other players. If a player beats the champ, then that player is considered the new (current) champ. Each player in the competition challenges the current champ once in the fixed order $e_1,e_2,\ldots,e_n$. The champ of the last round is considered the winner of the tournament. 
When we consider the possibility of manipulation in tournaments, we are supposed to possess information about the probabilities of the outcomes of the matchings. Here, the standard probabilistic model is to assume that for each pair of players that can potentially compete against each other, we know the probability of one of them beating the other (and, hence, we also know the probability of the other beating the first); see, e.g., \cite{mattei2015complexity,kim2015fixing,vu2008agenda,aziz2014fixing,stanton2011manipulating}. Constructive bribery is the most ubiquitous form of manipulation in both competition and voting~\cite{mattei2015complexity,saarinen2015probabilistic,faliszewski2006complexity,faliszewski2009hard,karia2023hard,tao2023electoral}, and its objective is to manipulate the selection process so that our favorite player/candidate wins. Here, we are often supposed to have a price associated with each possible bribing action along with a budget.

Accordingly, in \probname  we are given, along with player set $\{e_{1},\ldots,e_{n},e^{*}\}$:
\begin{itemize}
\item For every player $e_i$, a {\em bribe vector}, which is a vector of price-probability pairs; each pair specifies the price of the bribe(s) required to make $e_i$ lose against $e^{*}$ with the specified probability. We can suppose that the vector includes a pair with price $0$, which corresponds to the probability of $e_i$ losing when no bribe is involved.
\item A budget $B\in \mathbb{N}$.
\item A threshold probability $t\in [0,1]$.
\end{itemize}

The rationale behind having a vector with more than two entries (and, in particular, losing probabilities other than $1$ when a bribe is involved) is that various ways can affect the probability of a team or player losing, each having a different price. For example, we can bribe a different number of players in $e_i$ (when $e_i$ is a team), different coaches of $e_i$, the judge(s) of that specific match, alter various environmental conditions (e.g., which player plays in which court), and more.

The goal of \probname  (formally defined in \cref{sec:prelim}) is to determine whether the probability of~$e^{*}$ winning the tournament can be increased to or above $t$ using bribes for the matches between~$e^*$ and $e_{1},\ldots,e_{n}$ based on their respective bribe vectors, without exceeding the budget $B$. We remark that our model is slightly more general than the one of \citet{mattei2015complexity}, since they require the probabilities to be encoded in unary and we do not.


The initial work of \citet{mattei2015complexity} proved the following results related to \probname:
\begin{itemize}
    \item  \probname belongs to NP. This follows from~\cite[Corollary~4.4]{mattei2015complexity}.
   
    \item \probname is in XP\footnote{Standard terminology in parameterized complexity is defined in \cref{sec:prelim}.} when parameterized by the number of players. This result implicitly follows from \cite[Theorem~4.9 \& Corollary~4.10]{mattei2015complexity}, since the number of rounds of the tournament and the number of games in every round are upper-bounded by the number $n$ of players.
    \item \probname can be solved in $O(B^2n)$ time~\cite[Theorem~4.13]{mattei2015complexity}, 
showing that \probname is solvable in  pseudo-polynomial time. 

    \item They established (weak) NP-hardness of a variant of \probname, where all probabilities are expressed as (negative) powers of two \cite[Theorem~4.17]{mattei2015complexity}. Note that their reduction requires a compact representation of the probabilities. Hence, the problem they addressed is not a special case of \probname, and their reduction does not imply (weak) NP-hardness of \probname.

\end{itemize}

Cup tournaments are extremely popular in sports competitions~\cite{chaudhary2023make,suksompong2021tournaments,williams_moulin_2016,vu2008agenda,gupta2018winning}, voting~\cite{vu2009complexity,laslier1997tournament}, and decision making~\cite{brandt2007pagerank,rosen1985prizes}. 
Roughly speaking, a cup tournament is conducted in $\log_2 n$ rounds: in each round, the remaining players are paired up into matches, and the losers are knocked out of the tournament; when a single player remains, it is declared the winner.  (A formal definition is given in  \cref{sec:prelim}.)
Concerning \probnamecup, Mattei et al.~\cite{mattei2015complexity} established its classification within NP. Additionally, for the deterministic setting, they showed that \probnamecup can be efficiently solved in polynomial time using a dynamic programming algorithm. Furthermore, they introduced a variant of \probnamecup, termed \textsc{Exact Bribery}, where the goal is to precisely spend a budget of $B$, keeping other things the same. This variant was proven to be NP-complete.

\paragraph{Our Contribution.} We start with establishing the (classical) computational complexity of \probname. In \cref{sec:hardness} we show the following.
\begin{itemize}
    \item \probname is weakly NP-hard.
\end{itemize}
This motivates developing parameterized algorithms~\cite{downey2013fundamentals,CyganFKLMPPS15,niedermeier2006invitation} for the problem. We consider three parameters: the number of players, the number of distinct bribe values, and the number of distinct probability values. 

\smallskip\noindent
{\em Number of players.} Tournaments often involve relatively few players. For example, usually, Tennis tournaments involve around 128 players, and boxing championships involve around 30 players in a weight category. Hence, the number of players is a highly practical parameter. However, in \cref{sec:hardness} we show the following
\begin{itemize}
    \item \probname is W[1]-hard when parameterized by the number $n$ of players.
\end{itemize}
This implies that the XP-algorithm by \citet{mattei2015complexity} presumably cannot be improved to an FPT-algorithm. To prove the result, we also show that the \PK problem is W[1]-hard when parameterized by the number of items in the knapsack. We believe this is a valuable result in its own right: \PK can be a useful source problem for reductions to additional problems in social choice that concern probabilities (and, hence, a product of numbers) as well. Moreover, we show that our result for \probname further implies that {\sc Constructive Bribery for Cup Tournaments} is W[1]-hard, too, when parameterized by the number of players. We believe that our aforementioned implication nicely extends the results of the literature on cup tournaments, which is abundant with studies of various forms of manipulation (including bribery) from the perspective of parameterized complexity~\cite{gupta2018winning,zehavi2023tournament,konicki2019topics,ramanujan2017rigging,aziz2014fixing}.

\smallskip

\noindent\emph{Number of distinct bribe and probability values.} The number of distinct bribes is a well-motivated parameter: often, prices of the same action do not have that many possibilities. For example, a certain judge or coach will ask for the same price (or have a small range of prices) irrespective of the player or team involved. Furthermore, it is conceivable that if, say, the budget is thousands of dollars, then bribes that are not multiplications of one thousand will not be discussed---this, too, reduces the number of distinct bribes possible. Moreover, the number of distinct probabilities is likely to be small as well. Probabilities are, essentially,  rough estimations, and hence, even if we have a wide range of them, they can be rounded up to the closest value from a (predetermined) small-sized set of probabilities. In \cref{sec:algorithmic}, we show that these parameters yield tractability. 
\begin{itemize}
    \item \probname in FPT when parameterized by the number of distinct bribe values.
    \item \probname in FPT when parameterized by the number of distinct probability values.
\end{itemize}
Both algorithms exhibit similarities and are derived through mixed integer linear program (MILP) formulations for \probname. 
To obtain the results, we develop a novel method of designing MILPs that are guaranteed to have optimal solutions where \emph{all} variables are set to integer values.
We believe this technique can be useful in various application areas and hence is of independent interest.

\paragraph{Application to Campaign Management.} A deeper look into the definition of the \probname\ problem shows that the order of the players $e_1,e_2,\ldots,e_n$ is irrelevant to its answer---i.e., if we reorder them, we obtain an equivalent instance in terms of whether the answer to our specific objective is yes or no. (Of course, reordering might affect $e_1,e_2,\ldots,e_n$, but not $e^*$, who has to play and win against all of them.) Thus, we can, essentially, suppose that $e_1,e_2,\ldots,e_n$ are unordered. This gives rise to other applications of our results, e.g., to the area of {\em campaign management}~\cite{bredereck2016large,elkind2009swap,elkind2010approximation}. Specifically, we can think of $e^*$ as a candidate (person, idea, or product) that aims to win the election/be approved, and of each $e_i$ as a voter (possibly representing a group of individuals who cast a single vote), whose support/consent is essential to $e^*$. Then, for each $e_i$, a price-probability pair represents an amount of money to invest in winning $e_i$'s support/consent (e.g., by advertising) and the estimated probability of that amount being enough.

\paragraph{Other Related Works.} 
In the standard \textsc{Tournament
Fixing Problem} (TFP) with input $(e^{*},P,t)$, where $t\in[0,1]$, $e^{*}\in [n]$
is a ``favorite” player among the set of players $[n]$, and $P$ is an $n\times n$ matrix where the entry $P_{i,j}$ gives the probability that player $i$ beats player $j$, the question is whether a self-interested organizer can select a seeding for which $e^{*}$
wins the balanced cup tournament with probability at least $t$ \cite{vu2009complexity}.


 The restricted version of TFP, where $P$ has entries from $\{0,1\}$ only (the deterministic case), has also been studied and shown to be NP-hard by \citet{aziz2014fixing}. Later, a substantial body of works investigated the parameterized complexity of deterministic TFP (see, e.g., \cite{zehavi2023tournament,ramanujan2017rigging,gupta2018winning,gupta2019succinct,roy2020select}).
 \citet{konicki2019bribery} explored a variation of TFP, allowing organizers to both arrange seeding and bribe players to decrease their probability of winning against others at a specified cost, provided it stays within a budget. In their model, the probability matrix $P$ is either deterministic or $\varepsilon$-monotonic, reflecting player ordering and constraints on winning probabilities.

A paper by \citet{saarinen2015probabilistic} proved that it is $\#$P-hard to manipulate
a round-robin tournament by controlling the outcome of a subset of the games to raise the probability of winning above a particular threshold. The result holds in the restricted case where all
probabilities are zero, one half, or one. 


\section{Problem Setting and Preliminaries} \label{sec:prelim}


In the setting of challenge the champ tournaments, we have $n+1$ players, say, $\{e_1,\ldots, e_n, e^*\}$. Player~$e^*$ is initially the champ. In each of the $n$ rounds, player $e_i$ challenges the current champ and is considered the new champ if they win the challenge. 
We formally define a challenge the champ tournament as follows. A tournament tree for a challenge the champ tournament is visualized in \cref{fig1}.

\begin{definition} [Challenge the Champ Tournament]
A \emph{challenge the champ tournament} consists of a set of $n+1$ players $\{e_1,\ldots, e_n, e^*\}$ and has $n$ rounds. Initially, player $e^*$ is considered the champ (of round 0). In round $i>0$, player $e_i$ challenges the champ of round $i-1$, say player $e$. If~$e_i$ beats $e$, then $e_i$ becomes the champ of round $i$. Otherwise, $e$ is the champ of round $i$. 
The champ of round $n$ is considered the \emph{winner} of the tournament.
\end{definition}

  \begin{figure}[t]
 \centering
    \includegraphics[scale=.8]{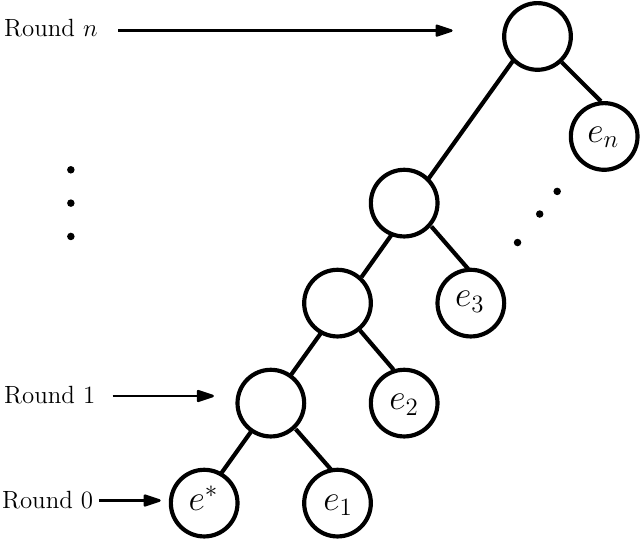}
    \caption{Illustration of a challenge the champ tournament with players $\{e_{1},\ldots,e_{n},e^{*}\}$. Here, $e^{*}$ is the initial champ.} 
    \label{fig1}
\end{figure}

\paragraph{Constructive Bribery.}
In this paper, we investigate the setting where players from the set $\{e_1,\ldots, e_n\}$ can be influenced through bribes to reduce their chances of winning against $e^*.$ Our objective is to determine if a specific budget for bribes can be allocated in such a way that player $e^*$ maintains the champ title and wins the tournament with a predefined probability, called \emph{threshold value} $t$, after facing each player in $\{e_1,\ldots, e_n\}$. Note that since we are only interested in cases where $e^*$ wins all games, the round in which each challenger plays against the champ does not matter. 

To formalize the problem, we introduce a so-called \emph{bribe vector} for each player in $\{e_1,\ldots, e_n\}$. We denote by $C_i$ the bribe vector for player $e_i$. Intuitively, $C_i$ specifies how much it costs to bribe players into lowering their winning probability against $e^*$.

\begin{definition}[Bribe Vector]
Let $e_i\in \{e_1,\ldots, e_n\}$ be a player and let $\ell_i\in \mathbb{N}$ be the number of different bribes that $e_i$ accepts. Then, the \emph{bribe vector} $C_i\in(\mathbb{N}\times [0,1])^{\ell_i}$ is a vector of length~$\ell_i$ with elements from $\mathbb{N}\times [0,1]$. Each element $C_i[j]=(b_j,p_j)\in \mathbb{N}\times [0,1]$ implies that bribing player~$e_i$ with amount $b_j$ increases their losing probability when playing against $e^*$ to $p_j$. We call~$b_j$ a \emph{bribe value} and $p_j$ a \emph{probability value}. Furthermore, we require for all $i$ that $C_i[1]=(0,p_1)$, that is, the first entry of each bribe vector contains the losing probability when playing against $e^*$ when no bribes are used. Moreover, we require that if $j<j'$ then $b_j<b_{j'}$, where $C_i[j]=(b_j,p_j)$ and~$C_i[j']=(b_{j'},p_{j'})$. 
\end{definition}

Now, given a set of bribes $\{j_1,\ldots,j_n\}$ for the players $\{e_1,\ldots,e_n\}$ (where $j_i=1$ if player $e_i$ is not bribed), the total cost of the bribes is $\sum_i b_{j_i}$, and the winning probability of $e^*$ is $\prod_i p_{j_i}$, where $(b_{j_i},p_{j_i})=C_i[j_i]$.
The main problem that we study in this paper is formally defined as follows.


\problemdef{\probnamelong (\probname)}
{A set of players $P=\{e_1,\ldots, e_n, e^*\}$, a bribe vector $C_i\in(\mathbb{N}\times [0,1])^{\ell_i}$ for each player $e_i\in P$, a probability threshold~$t\in[0,1]$, and a budget $B\in \mathbb{N}$.} 
{Can we raise $e^{*}$'s probability of winning the challenge the champ tournament to at least $t$ by bribing the players $\{e_1,\ldots,e_n\}$ according to their bribe vectors and not exceeding the budget~$B$?}

Furthermore, we make the following observation about bribe vectors. We call a bribe vector $C$ \emph{monotone}, if for all $j<j'$ we have that $p_j<p_{j'}$, where $C[j]=(b_j,p_j)$ and $C[j']=(b_{j'},p_{j'})$.
    \begin{lemma}\label{lemma:values}
     Given an instance of \probname with bribe vectors ${C_i}$ with $i\in\{1,\ldots,n\}$, we can compute an equivalent instance in polynomial time with monotone bribe vectors.
    \end{lemma}  
    \begin{proof}
        Assume that there is a player $e_i\in\{e_1,\ldots,e_n\}$ such that $C_i$ is not monotone, that is, for some $1\le j<j'\le |C_i|$ we have $p_j\ge p_{j'}$, where $C_i[j]=(b_j,p_j)$ and $C_i[j']=(b_{j'},p_{j'})$. Let $j,j'$ be such that $j'-j$ is minimal. Note that we must have that $j=j'-1$.
        Then, we create a bribe vector ${C'_i}$ of length $|C_i|-1$ such that ${C'_i}[\ell]={C_i}[\ell]$ for all $1\le \ell\le j$ and ${C'_i}[\ell]={C_i}[\ell+1]$ for all $j<\ell<|C_i|$.

        We have that the \probname instance with the modified bribe vector is a yes-instance if and only if the original instance is a yes-instance. Let $C_i[j]=(b_j,p_j)$ and $C_i[j']=(b_{j'},p_{j'})$. 
        If there is a solution to the original instance that uses value $b_j$ to bribe player $e_i$ to have losing probability $p_j$, then this is also a valid solution to the modified instance. 
        If there is a solution to the original instance that uses value $b_{j'}$ to bribe player $e_i$ to have losing probability $p_{j'}$, then we can create a valid solution to the modified instance by bribing player $e_i$ with value $b_j$ to have losing probability~$p_j$. Since $b_j<b_{j'}$ and $p_j\ge p_{j'}$, the budget is not violated and the winning probability of $e^*$ is not decreased.
        If there is a solution to the modified instance, then this solution is clearly also valid for the original instance.

        By repeating the described procedure, we can create modified bribe vectors with the desired property such that the \probname instance with the modified bribe vectors is a yes-instance if and only if the original instance is a yes-instance.
    \end{proof}

Due to \cref{lemma:values}, we from now on assume without loss of generality that the bribe vectors of every \probname instance are monotone.

\paragraph{\probnamecup.}

We have expanded the concept of constructive bribery to another well-known tournament, called \emph{cup tournament} (also known as  \emph{knockout tournaments}). In cup tournaments, a set of $n$ \emph{players} is provided (where, for simplicity, $n$ is a power of $2$), along with a {\em seeding} that dictates how to label the $n$ leaves of a complete binary tree representing the players. Given a seeding, the competition is conducted in rounds as follows. As long as the tree has at least two leaves, every two players with a common parent in the tree play against each other, and the winner is promoted to the common parent; then, the leaves of the tree are deleted from it. This process continues until only one player remains, who is then declared the \emph{winner} of the tournament. 

Here also, the objective is to manipulate the players by offering bribes, under a given budget. Based on their assigned bribe vectors, the players can decrease their winning probability, ensuring that a designated favorite player, say $e^{*}$, emerges as the winner with a probability of at least a given threshold. It is essential to note that, unlike in \probname, matches can occur among players who are not favorites. Consequently, the bribe vectors are defined not only in relation to the favorite player but also among the players themselves.

More formally, the \probnamecup is defined as follows.

\problemdef{\probnamecup}
{A set $P$ of players with $|P|=n=2^{n'}$ for some $n'\in\mathbb{N}$, a favorite player $e^{*}\in P$, a seeding $\sigma$ of $P$, a bribe vector $C_i^{j}\in(\mathbb{N}\times [0,1])^{\ell_i}$ for each player $i$ in $P$ against $j$, a probability threshold~$t\in[0,1]$, and a budget $B\in \mathbb{N}$.} 
{Can we raise $e^{*}$'s probability of winning the cup tournament to at least $t$ by bribing the players in $P\setminus\{e^{*}\}$ according to their bribe vectors and not exceeding the budget~$B$?}

\paragraph{Parameterized Complexity.}
 We use the standard concepts and notations from parameterized complexity theory~\cite{downey2013fundamentals,CyganFKLMPPS15,niedermeier2006invitation}.
A \emph{parameterized problem} $L\subseteq \Sigma^{*} \times \mathbb{N}$ is a subset of all instances $(x, k)$ from $\Sigma^{*} \times \mathbb{N}$, where $k$ denotes the \emph{parameter}. A parameterized problem $L$ is in the complexity class XP if there is an algorithm that solves each instance $(x,k)$ of $L$
 in $x^{f(k)}$ time, for some computable function $f$. Furthermore, $L$ is in the class FPT (or fixed-parameter tractable), if there is an algorithm that decides every instance $(x, k)$ for $L$ in $f(k)\cdot|x|^{
O(1)}$ time, where $f$ is any
computable function that depends only on the parameter. If a
parameterized problem $L$ is W[1]-hard, then it is presumably
not fixed-parameter tractable. 
		 

	

\section{Hardness Results} \label{sec:hardness}
In this section, we present our computational hardness results. In particular, we show that \probname is weakly NP-hard and W[1]-hard when parameterized by the number of players. In particular, our results imply that the XP-algorithm for \probname parameterized by the number of players by \citet{mattei2015complexity} presumably cannot be improved to an FPT-algorithm.

\paragraph{Parameterized Hardness of Product Knapsack.}
To obtain our computational hardness result for \probname, we provide a reduction from the so-called \MPK problem. To this end, we first show a parameterized hardness result for \PK (the non-colored version of \MPK).

\problemdef{\PK}
{Items $j\in N:= \{1, \ldots , n\}$ with weights $w_j \in \mathbb{N}$ and profits $v_j \in \mathbb{N}$, a positive knapsack capacity $C \in \mathbb{N}$, and a value $V\in \mathbb{N}$.}
{Does there exist a subset $S\subseteq N$ with $\displaystyle\sum_{j\in S}w_{j}\leq C$ such that $\displaystyle\prod_{j\in S}v_{j}\geq V$?}

  
  
       \PK is known to be weakly NP-hard~\cite{pferschy2021approximating,halman2019bi}. We show that \PK is W[1]-hard when parameterized by the number of items in the knapsack. This result is of independent interest and allows us to obtain parameterized hardness results from our reduction.
  
\begin{theorem} \label{thm:w1}
\PK
is W[1]-hard when parameterized by the number of items in the knapsack.    
\end{theorem}
\begin{proof}
    We modify the reduction by \citet{halman2019bi} that they used to show weakly NP-hardness for \PK to show the W[1]-hardness of \PK when parameterized by the number of items in the knapsack. 
    Instead of reducing from a variant of the \textsc{Partition} problem, as done by \citet{halman2019bi}, we reduce from \textsc{Small $k$-Sum} parameterized by $k$, which is known to be W[1]-hard~\cite{abboud2014losing}. Here, we are given a set of numbers $S=\{s_1, s_2, \ldots, s_n\}$ with $s_i\in[-n^{2k},n^{2k}]$, and we are asked whether there exists $S^\star\subseteq S$ with $|S^\star|=k$ such that $\sum_{s\in S^\star}s=0$.

    Given a \textsc{Small $k$-Sum} instance, we first slightly modify the instance to obtain an instance of a slightly modified version of \textsc{Small $k$-Sum}. Let $s_i'=s_i+2n^{2k}+n^{k^2}$ for all $i\in[n]$ and let $S'=\{s_1',s_2',\ldots, s_n'\}$. Clearly, we have $s_i'\in[n^{2k}+n^{k^2},3n^{2k}+n^{k^2}]$. Now we ask whether there exists $S^\star\subseteq S'$ with $|S^\star|=k$ such that $\sum_{s\in S^\star}s=k(2n^{2k}+n^{k^2})$.
    It is easy to observe that the original \textsc{Small $k$-Sum} instance is a yes-instance if and only if the modified instance is a yes-instance of the modified version of \textsc{Small $k$-Sum}. Furthermore, clearly, the modified version of \textsc{Small $k$-Sum} is also W[1]-hard when parameterized by $k$.

    From now on, assume we are given an instance $S'$ of the modified version of \textsc{Small $k$-Sum}. Denote $T=k(2n^{2k}+n^{k^2})$. We construct an instance of PKP as follows. For each number $s_i'\in S'$ we create an item $i$ with
    \begin{itemize}
        \item $w_i=s_i'$, and
        \item $v_i=(1-\frac{s_i'}{T^2})^{-1}$.
    \end{itemize}
    
    Furthermore, we set $C=T$ and $V=(1-\frac{1}{T}+\frac{1}{2T^2})^{-1}$. Note that the values $v_i$ and the lower bound on the value product $V$ are not necessarily integers, but we can multiply each value with an appropriate number $N$ and multiply $V$ with $N^k$ to obtain integer values.

    The reduction can be computed in polynomial time. It is easy to see that at most $k$ items fit into the knapsack, since the weight of any $k+1$ items is larger than $C$. As we will show in the correctness proof, at least $k$ items are also necessary to obtain a value of at least $V$.

    $(\Rightarrow)$: Assume the modified \textsc{Small $k$-Sum} instance is a yes-instance. Then there is a subset $S^\star\subseteq S'$ with $|S^\star|=k$ such that $\sum_{s\in S^\star}=T$. 
    We remark that this direction of the correctness proof is very similar to the one provided by \citet{halman2019bi}.
    We claim that the items corresponding to the numbers in $S^\star$ form a solution to the \PK instance. 
    
    Let $I^\star = \{i\mid s_i'\in S^\star\}$. First, observe that we have
    \[
    \sum_{i\in I^\star} w_i=\sum_{s\in S^\star} s = T \le C.
    \]
    Hence, the items in $I^\star$ fit into the knapsack.
    
    Next, we show that the items in $I^\star$ have a sufficiently large value.
    We start with the following observations.
    \[
    \prod_{i\in I^\star} v_i = \prod_{s\in S^\star} (1-\frac{s}{T^2})^{-1}\ge V = (1-\frac{1}{T}+\frac{1}{2T^2})^{-1} \Longleftrightarrow \prod_{s\in S^\star} (1-\frac{s}{T^2})\le 1-\frac{1}{T}+\frac{1}{2T^2}
    \]
    By taking the logarithm on both sides of $\prod_{s\in S^\star} (1-\frac{s}{T^2})\le 1-\frac{1}{T}+\frac{1}{2T^2}$ we get
    \begin{equation}\label{eq:1}
    \sum_{s\in S^\star}\ln(1-\frac{s}{T^2})\le \ln(1-\frac{1}{T}+\frac{1}{2T^2}).
    \end{equation}
    
    Note that $0<1-\frac{s}{T^2}<1$ for all numbers $s\in S'$ and $0<1-\frac{1}{T}+\frac{1}{2T^2}<1$. Hence, we can take the Maclaurin series of the logarithms to obtain
    \[
    \ln(1-\frac{s}{T^2})=-\sum_{i=1}^\infty \frac{1}{i}(\frac{s}{T^2})^i,
    \]
    and
    \[
    \ln(1-\frac{1}{T}+\frac{1}{2T^2}) = -\sum_{i=1}^\infty \frac{1}{i}(\frac{1}{T}-\frac{1}{2T^2})^i.
    \]
    
    Using the Maclaurin series of the logarithms in Equation~\ref{eq:1}, we obtain
    \[
    \sum_{s\in S^\star}\sum_{i=1}^\infty \frac{1}{i}(\frac{s}{T^2})^i\ge  \sum_{i=1}^\infty \frac{1}{i}(\frac{1}{T}-\frac{1}{2T^2})^i.
    \]
    
    Now we multiply both sides of the inequality by $T^2$ and split off some summands on both sides to obtain
    \[
    \sum_{s\in S^\star} s + \sum_{s\in S^\star}\sum_{i=2}^\infty \frac{1}{i}\frac{s^i}{T^{2i-2}}\ge T-\frac{1}{2T}+\frac{1}{8T^2} + \sum_{i=3}^\infty \frac{1}{i}\frac{(2T-1)^i}{2^iT^{2i-2}}.
    \]
    
    It follows that Equation~\ref{eq:1} is equivalent to
    \begin{equation}\label{eq:2}
    \sum_{s\in S^\star} s \ge T-\frac{1}{2T}+\frac{1}{8T^2} + \sum_{i=3}^\infty \frac{1}{i}\frac{(2T-1)^i}{2^iT^{2i-2}}-\sum_{s\in S^\star}\sum_{i=2}^\infty \frac{1}{i}\frac{s^i}{T^{2i-2}}
    \end{equation}
    
    Since $S$ is a solution to the modified \textsc{Small $k$-Sum} instance, we have that $\sum_{s\in S^\star} s=T$. Hence, it remains to show that 
    \[
    -\frac{1}{2T}+\frac{1}{8T^2} + \sum_{i=3}^\infty \frac{1}{i}\frac{(2T-1)^i}{2^iT^{2i-2}}-\sum_{s\in S^\star}\sum_{i=2}^\infty \frac{1}{i}\frac{s^i}{T^{2i-2}}\le 0.
    \]
    
    We can assume w.l.o.g.\ that $T\ge 4$ and we have $0<\frac{1}{T}<1$. Hence, we get
    \[
    \sum_{i=3}^\infty \frac{1}{i}\frac{(2T-1)^i}{2^iT^{2i-2}}\le \frac{1}{3} \sum_{i=3}^\infty \frac{(2T)^i}{2^iT^{2i-2}}=\frac{1}{3}\sum_{i=1}^\infty\frac{1}{T^i}=\frac{1}{3(T-1)}.
    \]
    
    It is easy to see that $-\frac{1}{2T}+\frac{1}{8T^2}+\frac{1}{3(T-1)}<0$ for $T\ge 4$ and we can conclude that Equation~\ref{eq:2} is satisfied. It follows that $\prod_{i\in I^\star} v_i\ge V$.
    
    $(\Leftarrow)$: Assume the constructed instance of \PK is a yes-instance and we have a solution $I^\star$. We claim that the numbers corresponding to the items in $I^\star$ form a solution to the modified \textsc{Small $k$-Sum} instance. Let $S^\star=\{s_i'\mid i\in I^\star\}$.
    First, observe that we have 
    \[
    \sum_{s\in S^\star}s=\sum_{i\in I^\star} w_i \le C = T.
    \]
    
    Furthermore, we have that $|S^\star|\le k$, since the weight of any $k+1$ items is at least $(k+1)n^{k^2}$ which is larger than $T$ for sufficiently large $k$.

    It remains to show that $\sum_{s\in S^\star}s\ge T$ and $|S^\star|\ge k$. We show that $\sum_{s\in S^\star}s\ge T$ which together with the observation above implies that $\sum_{s\in S^\star}s= T$. This immediately also implies that $|S^\star|\ge k$ since the sum of any $k-1$ numbers in $S'$ is strictly smaller than $T$ for sufficiently large $k$.
    Since all numbers in $S'$ are integers, showing $\sum_{s\in S^\star}s\ge T$ is equivalent to showing $\sum_{s\in S^\star}s> T-1$.
    Furthermore, we have that 
    \[
    \prod_{i\in I^\star} v_i = \prod_{s\in S^\star} (1-\frac{s}{T^2})^{-1}\ge V.
    \]
    
    From Equation~\ref{eq:2} it follows that $\sum_{s\in S^\star}s> T-1$ is equivalent to
    \[
    -\frac{1}{2T}+\frac{1}{8T^2} + \sum_{i=3}^\infty \frac{1}{i}\frac{(2T-1)^i}{2^iT^{2i-2}}-\sum_{s\in S^\star}\sum_{i=2}^\infty \frac{1}{i}\frac{s^i}{T^{2i-2}}>-1.
    \]
    
    For the next step, observe that for all $s\in S'$ we have that $s\le \frac{2}{k}T$ and $|S^\star|\le k$. Hence, we get
    \[
    \sum_{s\in S^\star}\sum_{i=2}^\infty \frac{1}{i}\frac{s^i}{T^{2i-2}}<k\sum_{i=2}^\infty \frac{1}{i}\frac{(\frac{2}{k}T)^i}{T^{2i-2}}=k\sum_{i=2}^\infty \frac{2^i}{ik^iT^{i-2}}<\frac{2}{k}+\frac{8}{3k^2}\sum_{i=1}^\infty \frac{2^i}{T^i}=\frac{2}{k}+\frac{16}{3k^2(T-1)}.
    \]
    
    We can assume w.l.o.g.\ that $k\ge 4$, then the maximum value of $\frac{2}{k}+\frac{16}{3k^2(T-1)}$ is attained at $k=4$ and we have
    \[
    \frac{2}{k}+\frac{16}{3k^2(T-1)}\le \frac{1}{32}+\frac{1}{12(T-1)}.
    \]
    
    It follows that
    \[
    -\frac{1}{2T}+\frac{1}{8T^2} + \sum_{i=3}^\infty \frac{1}{i}\frac{(2T-1)^i}{2^iT^{2i-2}}-\sum_{s\in S^\star}\sum_{i=2}^\infty \frac{1}{i}\frac{s^i}{T^{2i-2}}>-\frac{1}{2T}+\frac{1}{8T^2}-\frac{1}{32}-\frac{1}{12(T-1)}.
    \]
    
    It remains to show that 
    \[
    -\frac{1}{2T}+\frac{1}{8T^2}-\frac{1}{32}-\frac{1}{12(T-1)}>-1,
    \]
    
    which is easy to verify. We can conclude that  $\sum_{s\in S^\star}s= T$ and hence also that $|S^\star|=k$.
\end{proof}


In the multicolored version of \PK, each item is assigned to a specific color class, and the objective is to select precisely one item from each color class to fill our knapsack. More formally, it is defined as follows.
\problemdef{\MPK}
{Items $j\in N:= \{1, \ldots , n\}$ with weights $w_j \in \mathbb{N}$ and profits $v_j \in \mathbb{N}$, a partition of $N$ into $k (\in \mathbb{N})$ sets $X_{1},\ldots,X_{k}$, a positive knapsack capacity $C \in \mathbb{N}$, and a value $V\in \mathbb{N}$.}
{Does there exist a subset $S\subseteq N$ containing exactly one item from each $X_{i}$ with $\displaystyle\sum_{j\in S}w_{j}\leq C$ such that $\displaystyle\prod_{j\in S}v_{j}\geq V$?}

In the realm of parameterized complexity, when a problem is parameterized by the solution size, it is well-known that by applying the color-coding technique (see \cite{CyganFKLMPPS15}), we can obtain a parameterized reduction (one-to-many) from the original problem to its multicolored counterpart parameterized by the number of colors. Thus, through a straightforward parameterized reduction, we get the following corollary of \cref{thm:w1}.
\begin{corollary}\label{thm:mpkw1h}
    \textsc{Multicolored} \PK is W[1]-hard when parameterized by the number of colors.
\end{corollary}

 \paragraph{Hardness of \probname.}
Using the previous results, in particular, \cref{thm:mpkw1h}, we now proceed to establish our main hardness result.
\begin{theorem} \label{thm:weakly}
\probname is weakly NP-hard and W[1]-hard when parameterized by the number of players.
\end{theorem}

\begin{proof}
We present a parameterized polynomial-time reduction from \textsc{Multicolored} \PK parameterized by the number of colors to \probname parameterized by the number of players. By \cref{thm:mpkw1h}, we have that \MPK is W[1]-hard when parameterized by the number of colors.
    
Let $\psi=(X_{1},\ldots,X_{k},C,V)$ be a given instance of \textsc{Multicolored} \PK, where $v_{1}^{i},\ldots,v_{|X_{i}|}^{i}$ denote the profits and $w_{1}^{i},\ldots,w_{|X_{i}|}^{i}$ denote the weights of the items present in the color class $X_{i}$ for every $i\in[k]$. Here, we assume without loss of generality that $w_{1}^{i}\leq w_{2}^{i} \leq \cdots \leq w_{|X_{i}|}^{i}$. Now, we construct an instance $\phi=(P,\overline{C},t,B)$ of \probname with $k+1$ players, where $P=\{e_1,\ldots,e_k,e^*\}$ and $\overline{C}$ contains the bribe vectors for every player in $P\setminus \{e^*\}$, as follows: 
    
    Player $e^{*}$ is the initial champ. Here, we informally say that for each $i\in [k]$, the player $e_i$ corresponds to the color class $X_{i}$. Also, for each $i\in [k]$, we set the bribe vector $C_{i}$ corresponding to the items in the color class $X_{i}$. Formally, the $j$th entry of the vector $C_{i}$ is defined as $C_{i}[j]=(w_{j}^{i},v_{j}^{i})$. Here, note that the length of the vector $C_{i}$ is $|X_{i}|$. Finally, we set budget $B=C$, and we set the threshold value $t=V$.


   This finishes the construction, which can clearly be done in polynomial time. Note that the number of players in the constructed instance is $k+1$.
    Now, we claim that $\psi$ is a yes-instance of \textsc{Multicolored} \PK if and only if $\phi$ is a yes-instance of \probname.

   $(\Rightarrow)$: Assume that $\psi$ is a yes-instance of  $\textsc{Multicolored}$ \PK. Let $S$ be a solution of $\psi$ such that  $\sum_{j\in S}w_{j}\leq C$ and $\prod_{j\in S}v_{j}\geq V$. Now, we bribe the players in $P\setminus \{e^*\}$ as follows.
   Without loss of generality, assume that $w_{j}$ and $v_{j}$ correspond to the weight and price of the item in the solution that is present in the color class $X_{j}$. Then by construction, we have that for player $e_{j}\in P\setminus\{e^{*}\}$ there must be some entry in the bribe vector $C_{j}$ that corresponds to the pair $(w_{j},v_{j})$. We bribe player $e_j$ with value $w_j$ to lower their winning probability against $e^*$ to~$v_j$. By construction, we do not exceed the budget, since $C=B$. Furthermore, the product of the winning probabilities $v_j$ of $e^*$ against players $e_j$ equals at least the threshold $t=V$.

    $(\Leftarrow)$: Assume that $\phi$ is a yes-instance of \probname. Hence, there exists a set of bribes of total cost at most $B=C$, such that the probability that $e^{*}$ wins the tournament is at least $t=V$. Now, for every $i\in [k]$, there must be an entry in $C_{i}$ that specifies how player $e_i$ is bribed. Let player $e_i$ be bribed according to $C_i[j]=(w^i_{j},v^i_{j})$. We put the corresponding item in $X_{i}$ into the knapsack. Thus, we have constructed a solution for \textsc{Multicolored} \PK of total weight at most $C$, where the product of the values is at least $V$ since it equals the winning probability of player $e^*$.
\end{proof}

\paragraph{Parameterized Hardness of \probnamecup}
Finally, the following result establishes that when the parameter is the number of players, the existence of an FPT algorithm remains unlikely for \probnamecup as well.

\begin{theorem} \label{corr:whard}
  \probnamecup is W[1]-hard when parameterized by the number of players.
\end{theorem}

\begin{proof}
We present a parameterized polynomial-time reduction from \probname to \probnamecup, with both problems being parameterized by the number of players. By \cref{thm:weakly}, we know that \probname is W[1]-hard when parameterized by the number of players.

 Let $\psi=(P,C,t,B)$ be an instance of \probname with $n$ players, where $P=\{e_1,\ldots,e_{n-1},e^*\}$ and $C$ contains the bribe vectors for every player in $P\setminus \{e^*\}$. Now, we construct an instance $\phi=(P',\widehat{e},D,t',B',\sigma)$ of \probnamecup with $2^{n-1}$ players, where $P'=\{e'_1,\ldots,e'_{n-1},\widehat{e},d_{1},\ldots,d_{2^{n-1}-n}\}$. Here, we refer to $\{e'_1,\ldots,e'_{n-1}\}$ as the \emph{main players}, and  $\{d_{1},\ldots,d_{2^{n-1}-n}\}$ as the \emph{dummy players}. Let $B'=B$ and $t'=t$. Also, let $D$ contain the bribe vectors for each player in $P'\setminus \{\widehat{e}\}$, structured as follows: The bribe vector for $e'_{i}$ against $\widehat{e}$ in $D$ is same as the bribe vector for $e_{i}$ against $e^{*}$ in $C$ for every $i\in [n-1]$. The remaining bribe vectors (for dummy players against every other player and for main players against everyone except $\widehat{e}$) have only one entry, indicating the probability of losing against the corresponding player. 
 Additionally, the values of bribe vectors are set in such a way that every main player defeats the dummy players with a probability of 1, and the remaining players can defeat each other with any probability $\frac{1}{2}$. Finally, the tournament seeding is defined as follows: $\widehat{e}$ is seeded in position 1, and $e'_{i}$ is seeded in position $2^{i-1}+1$, where $i\in [n-1]$. The dummy players are arbitrarily seeded in the remaining positions.
  
This finishes the construction. Now, we claim that $\psi$ is a yes-instance of \probname if and only if~$\phi$ is a yes-instance of \probnamecup.

   $(\Rightarrow)$: Assume that $\phi$ is a yes-instance of \probname. Hence, there exists a set of bribes of total cost at most $B$, such that the probability that $e^{*}$ wins the tournament is at least $t$. Now, for every $i\in [n-1]$, there must be an entry in $C_{i}$ that specifies how player $e_i$ is bribed. Let player~$e_i$ be bribed according to $C_i[j]$. Now, with the specified seeding and probabilities in $\psi$, player $\widehat{e}$ is guaranteed to face player $e'_{i}$ in round $i$ for every $i\in [n-1]$. We make sure that the corresponding player $e'_{i}$ is bribed according to $C_i[j]$ against player $\widehat{e}$, i.e., $D_{i}^{\widehat{e}}[j]=C_{i}[j]$ value of the bribe vector is chosen for $e'_{i}$. Additionally, player $e'_{i}$ progresses to round $i$ with a probability of 1, without using anything from the budget. Consequently, the minimum winning probability for $\widehat{e}$ in $\psi$ is assured to be at least $t=t'$. Thus, we have constructed a solution for \probnamecup.
   
    $(\Leftarrow)$: Assume that $\psi$ is a yes-instance of \probnamecup. Hence, there exists a set of bribes with a total cost of at most $B'$, ensuring that the probability of $\widehat{e}$ winning the tournament is at least $t'$. Note that with the specified seeding and probabilities, player~$\widehat{e}$ is guaranteed to face player $e'_{i}$ in round $i$ for every $i\in [n-1]$. 
    Now, for every $i\in [n-1]$, there must be an entry in $D_{i}^{\widehat{e}}$ specifying how player $e'_i$ is bribed against $\widehat{e}$. 
Let player $e'_i$ be bribed according to $D_i^{\widehat{e}}[j](=C_{i}[j])$. We construct a solution for \probname by ensuring that the corresponding player~$e_{i}$ is bribed according to $C_i[j]$. Given that $B=B'$ and $t=t'$, we have constructed a solution for \probname.
\end{proof}



\section{Algorithmic Results} \label{sec:algorithmic}
In this section, we present our algorithmic results for \probname. We show two fixed-parameter tractability results. One is for the number of distinct bribe values as a parameter, and the other is for the number of distinct probability values as a parameter. Both algorithms are similar and are obtained by mixed integer linear program (MILP) formulations for \probname. 
\problemdef{\textsc{Mixed Integer Linear Program} (MILP)}{A vector $x$ of $n$ variables of which some are considered integer variables, a constraint matrix $A\in\mathbb{R}^{m\times n}$, two vectors $b\in\mathbb{R}^m$, $c\in \mathbb{R}^n$, and a target value $t\in \mathbb{R}$.}{Is there an assignment to the variables such that all integer variables are set to integer values, $c^\intercal x\ge t$, $Ax\le b$, and $x\ge 0$?}
Note that MILPs are also often considered to be optimization problems where instead of requiring $c^\intercal x\ge t$, the value of $c^\intercal x$ should be maximized.
MILPs are known to be solvable in FPT-time when the number of integer variables is the parameter~\cite{Lenstra1983Integer,dadush2011enumerative}. 
\begin{theorem}[\cite{Lenstra1983Integer,dadush2011enumerative}]\label{thm:MILP}
    MILP is fixed-parameter tractable when parameterized by the number of integer variables.
\end{theorem}

We build our MILP formulations in a specific way that ensures that there always exist optimal solutions where \emph{all} variables are set to integer values. To this end, we establish a general result concerning MILPs that, to the best of our knowledge, has not been employed before. 
While this result can straightforwardly be derived from known results, it may be of independent interest.
\begin{proposition}\label{prop:MILP}
    Let the following be an MILP.
    \[
    \max c^\intercal x \text{ subject to } Ax\le b, x\ge 0.
    \]
    Let $x = (x_{\text{int}} \ x_{\text{frac}})^\intercal$, where $x_{\text{int}}$ (resp.\ $x_{\text{frac}}$) denote the integer (resp.\ fractional) variables of the MILP. Let $A = (A_{\text{int}} \ A_{\text{frac}})$ where $A_{\text{int}}$ are the first $|x_{\text{int}}|$ columns of $A$, that is, the coefficients of the integer variables, and $A_{\text{frac}}$ are the remaining columns, that is, the coefficients of the fractional variables.
    If $A_{\text{frac}}$ is totally unimodular, then there exists an optimal solution to the MILP where all variables are set to integer values.
\end{proposition}
\begin{proof}
    Let $x^\star$ be an optimal solution to the MILP. Suppose that $A_{\text{frac}}$ (as defined in \cref{prop:MILP}) is totally unimodular. Let $c^\star$ denote the objective value achieved by $x^\star$. 
    Let $x^\star_{\text{int}}$ be the assignment to the integer variables, and let $x^\star_{\text{frac}}$ be the assignment to the fractional variables of the MILP in the optimal solution $x^\star$. Let $c=(c_{\text{int}} \ c_{\text{frac}})^\intercal$, where $c_{\text{int}}$ are the first $|x_{\text{int}}|$ entries of $c$, that is, the coefficients of the integer variables, and $c_{\text{frac}}$ are the remaining entries, that is, the coefficients of the fractional variables.
    Define the following linear program (LP):
    \[
    \max c_{\text{frac}}^\intercal x_{\text{frac}} \text{ subject to } A_{\text{frac}}x_{\text{frac}}\le \hat{b}, x_{\text{frac}}\ge 0,
    \]
    where $\hat{b}$ are the last $|x_{\text{frac}}|$ entries of $b-A_{\text{int}}x^\star_{\text{int}}$. Clearly, we have that $x^\star_{\text{frac}}$ is a feasible solution to the LP that achieves objective value $c_{\text{frac}}^\star= c^\star - c_{\text{int}}^\intercal x^\star_{\text{int}}$.

    It is well known that since $\hat{A}$ is totally unimodular, the LP admits an optimal solution where all variables are set to integer values~\cite{dantzig1956linear}. Let $x^{\star\star}_{\text{frac}}$ denote an optimal solution for the LP that sets all variables to integer values. Clearly, the achieved objective value of the solution $x^{\star\star}_{\text{frac}}$ is at least~$c^\star_{\text{frac}}$.

    Now, if we set the fractional variables of the MILP to $x^{\star\star}_{\text{frac}}$ (instead of $x^{\star}_{\text{frac}}$) and set the integer variables of the MILP to $x^{\star}_{\text{int}}$, we obtain a feasible solution to the MILP that achieves objective value at least $c^\star$. It is feasible since otherwise, a constraint in the LP must be violated. It has an objective value of at least $c^\star$ since the objective value achieved in the LP is at least $c^\star_{\text{frac}}=c^\star - c_{\text{int}}^\intercal x^\star_{\text{int}}$. The objective value achieved by the newly constructed solution to the MILP is also at most $c^\star$, since $c^\star$ is the objective value achieved by an optimal solution. We conclude that there exists an optimal solution to the MILP that sets all variables to integer values.
\end{proof}

\paragraph{Number of Distinct Bribe Values.}
Now, we are ready to state our results. Formally, we first prove the following.

\begin{theorem}\label{thm:numbervalues}
    \probname is fixed-parameter tractable when parameterized by the number of distinct bribe values.
\end{theorem}
To prove \cref{thm:numbervalues}, we provide a mixed integer linear program (MILP) formulation for \probname where the number of integer variables is upper-bounded by a function of the number of distinct bribe values in the \probname instance. 

\begin{proof}[Proof of \cref{thm:numbervalues}]
    We provide the MILP formulation for \probname where the number of integer variables is upper-bounded by a function of the number of distinct bribe values. Assume we are given an instance of \probname. We construct an MILP as follows.

    Let $v_{\#}$ denote the number of distinct bribe values and let $V$ denote the set of distinct bribe values.
    For each combination of a subset $V'\subseteq V$ and a value in that subset $v'\in V'$, we create an integer variable $x_{v',V'}$ that, intuitively, counts how many times we bribe a player with value $v'$ that has set of bribe values $V'$. 
    We call the set of probability values $P$ in the bribe vector of a player~$e_i$ the player's \emph{probability profile}.
    From \cref{lemma:values} follows that if two players $e_i, e_j$ have the same probability profile $P$ and have the same set of bribe values $V'$, we must have that $|P|=|V'|$ and hence bribing $e_i$ with some value $v'\in V'$ and bribing $e_j$ with the same value $v'$ increases the losing probability of $e_i$ and~$e_j$ to the same $p\in P$. We denote this probability with $p=p(P,v',V')$. In other words, players are uniquely characterized by their probability profile and their set of bribe values.

     For each combination of a subset $V'\subseteq V$, a value in that subset $v'\in V'$, and a probability profile $P$ (that appears in the \probname instance), we create a rational-valued variable~$x_{P,v',V'}$ that, intuitively, counts how many times a player that has set of bribe values $V'$ and probability profile $P$ is bribed with value $v'$ (to increase its losing probability to a uniquely determined $p=p(P,v',V')$).

     We want to maximize the following.
     \[
     \prod_{p} p^{\sum_{P,v',V'\mid p=p(P,v',V')}x_{P,v',V'}}
     \]
     
     This is equivalent to maximizing the logarithm of the expression. Hence, we have the following (linear) objective function.
     \[
     \sum_p \log p \left( \sum_{P,v',V' \mid p=p(P,v',V')}x_{P,v',V'} \right)
     \]

     We have the following constraints. The first one ensures that we do not violate the budget.
     \begin{equation}\label{constr:1}
         \sum_{v',V'} v' \cdot x_{v',V'}\le B
     \end{equation}
    The second set of constraints ensures that the number of times we use a value $v'$ to bribe a player that has the set of bribe values $V'$ (which is specified by $x_{v',V'}$) is the same as the sum of all times we use value $v'$ to bribe a player that has set of bribe values $V'$ and that has probability profile $P$.
    \begin{equation}\label{constr:2}
       \forall \ v',V' \colon \sum_{P} x_{P,v',V'} = x_{v',V'}
    \end{equation}
    The third set of constraints ensures that we do not use a value $v'$ to bribe a player that has the set of bribe values $V'$ and probability profile $P$ too many times. 
    Let $n_{P,V'}$ denote the number of players that have a set of bribe values $V'$ and the probability profile $P$. 
    \begin{equation}\label{constr:3}
        \forall \ P,V' \colon \sum_{v'}x_{P,v',V'} = n_{P,V'}
    \end{equation}
    Lastly, in the fourth set of constraints, we require that all fractional variables $x_{P,v',V'}$ are non-negative.
    \begin{equation}\label{constr:4}
        \forall \ P,v',V' \colon 0 \le x_{P,v',V'}
    \end{equation}

    It is easy to observe that the overall number of variables and constraints is in $2^{O(v_\#)}\cdot n$ whereas the number of integer variables is in $2^{O(v_\#)}$. By \cref{thm:MILP}, we can compute an optimal solution for the MILP in FPT-time with respect to the number $v_\#$ of distinct bribe values.

    In the remainder, we show that there is a solution to the MILP with
    \[
     \prod_{p} p^{\sum_{P,v',V'\mid p=p(P,v',V')}x_{P,v',V'}}\ge t
     \]
     if and only if the input instance of \probname is a yes-instance.

     $(\Rightarrow)$: Assume the input instance of \probname is a yes-instance. Then, it is possible to bribe players using budget $B$ such that the winning probability of $e^*$ is at least $t$. Let player $e_i$ be bribed with value $v_i$ in the solution and let the resulting losing probability of $e_i$ versus $e^*$ be $p_i$.

     We construct a solution for the constructed MILP as follows. Initially, we set all variables to~0. Now iterate over all players. Let player $e_i$ have a set of bribe values $V'$ and probability profile $P$. Then we increase the value of variable $x_{P,v_i,V'}$ by~1. Note that after his procedure, the constraints~(\ref{constr:3}) and~(\ref{constr:4}) are clearly met.

     Afterward, we set all integer values such that constraints~(\ref{constr:2}) are satisfied. Since these are equality constraints, this uniquely specifies how the integer variables are set. Furthermore, since we set all fractional variables to integer values in the previous step, we clearly also set all integer variables to integer values.

     Next, we argue that constraint~(\ref{constr:1}) is satisfied.
    To this end, note that every bribe with value $v'$ is accounted for exactly once by increasing the value of a variable $x_{P,v',V'}$ by~1. It follows that the same bribe is accounted for exactly once in the value of variable $x_{v',V'}$. Since the input instance is a yes-instance, the sum of all bribes is at most $B$, hence, we have that constraint~(\ref{constr:1}) is satisfied.

    Lastly, we argue that 
     \[
     \prod_{p} p^{\sum_{P,v',V'\mid p=p(P,v',V')}x_{P,v',V'}}\ge t.
     \]
    Similarly as in the argument before, note that every challenge that player $e^*$ can win with probability~$p$ is accounted for exactly once by increasing the value of a variable $x_{P,v',V'}$ such that $p=p(P,v',V')$ by~1. Hence, the number of challenges that $e^*$ can win with probability $p$ is $\sum_{P,v',V'\mid p=p(P,v',V')}x_{P,v',V'}$. Since the input instance is a yes-instance, player $e^*$ can win all challenges with a probability of at least $t$. It follows that the above inequality is fulfilled.
     
     $(\Leftarrow)$: Assume that we have a solution $x^\star$ to the created MILP instance such that 
     \[
     \prod_{p} p^{\sum_{P,v',V'\mid p=p(P,v',V')}x_{P,v',V'}}\ge t.
     \]
     In the following, we show how to bribe the players to increase the overall winning probability of $e^*$ to at least $t$.
     
     To this end, we show that there exists an optimal solution to the created MILP where \emph{all} variables are set to integer values. We do this using \cref{prop:MILP}. 

     Note that the constraint~(\ref{constr:1}) is independent of the fractional variables. Furthermore, constraints~(\ref{constr:3}) 
     are independent from the integer variables. We transform the constraints~(\ref{constr:2}) to constraints for the fractional variables by treating the integer variables as arbitrary constants. After that, we have a constraint matrix for the fractional variables consisting of the modified constraints~(\ref{constr:2}) and the constraints~(\ref{constr:3}). In the following, we will show that the corresponding constraint matrix is totally unimodular, which then by \cref{prop:MILP} implies that there exists an optimal solution to the MILP that sets all variables to integer values.

     First, note that the constraints~(\ref{constr:2}) partition the set of fractional variables, that is, each fractional variable is part of exactly one of the constraints~(\ref{constr:2}). We have the same for the constraints~(\ref{constr:3}). Furthermore, the coefficients in the constraint matrix for each variable are either 1 (if they are part of a constraint) or 0. It follows that the constraint matrix is a 0-1 matrix with exactly two 1's in every column. Additionally, in each column, we have that one of the two 1's appears in a row corresponding to the constraints~(\ref{constr:2}) and the other 1 is in a row corresponding to the constraints~(\ref{constr:3}). This is a sufficient condition for the constraint matrix to be totally unimodular~\cite{dantzig1956linear}.

     Thus, from now on, we can assume that the optimal solution $x^\star$ to the MILP sets all variables to integer values. We construct the bribes for the players as follows.
     Let $V'$ be a set of bribe values and~$P$ be a probability profile. Consider the set $E_{P,V'}$ of all players that have the set of bribe values~$V'$ and that have probability profile $P$. For each $v'\in V'$ we bribe $x_{P,v',V'}$ players of $E_{P,V'}$ with~$v'$. Note that constraints~(\ref{constr:3}) ensure that there are sufficiently many players in $E_{P,V'}$. Each bribe done this way is accounted for exactly once in the variable $x_{v',V'}$ due to the constraints~(\ref{constr:2}). Since constraint~(\ref{constr:1}) is satisfied, we have that the total amount of bribes does not exceed the budget~$B$. 
     
     It remains to show that the winning probability of $e^*$ after bribing the players is at least $t$. To this end, recall that a player with the set of bribe values $V'$ and probability profile $P$ loses with probability $p$ when playing against $e^*$ if and only if they are bribed with some $v'\in V'$ such that $p=p(P,v',V')$.
     It follows that the number of players that have losing probability $p$ when playing against $e^*$ is 
     \[
     \sum_{P,v',V'\mid p=p(P,v',V')}x_{P,v',V'}.
     \]
     Hence, we have that the overall winning probability of $e^*$ is 
    \[
     \prod_{p} p^{\sum_{P,v',V'\mid p=p(P,v',V')}x_{P,v',V'}},
     \]
     which by assumption is at least $t$.
\end{proof}

\paragraph{Number of Distinct Probability Values.}
With an analogous approach, we obtain the following result.

\begin{theorem}\label{thm:numberprobs}
    \probname is fixed-parameter tractable when parameterized by the number of distinct probability values.
\end{theorem}
\cref{thm:numberprobs} can be achieved by a similar MILP as the one presented to prove \cref{thm:numbervalues}. There are some key differences while the main idea remains the same. We omit a formal proof but describe the MILP below. The correctness of the MILP can be proven in an analogous way as in the proof of \cref{thm:numbervalues}. 

For every probability profile $P$ that appears in the \probname instance and every probability $p\in P$, we create an integer variable $x_{p,P}$ that, intuitively, counts how many players with probability profile~$P$ are bribed to lose against $e^*$ with probability $p$. Additionally, we have rational variables~$x_{p,P,V'}$ for every probability profile $P$, every probability $p\in P$, and every set of bribe values $V'$. These variables, intuitively, count how many players with probability profile $P$ and set of bribe values $V'$ are bribed to lose against $e^*$ with probability $p$. Note that due to \cref{lemma:values}, the bribe value to achieve this is uniquely determined, hence, we denote it with $v'=v(p,P,V')$.

Now, we want to minimize\footnote{Note that we introduced MILP as a maximization problem. However, we can easily switch between minimization and maximization by multiplying every entry of the vector $c$ by $-1$.} the following objective function, which quantifies the budget needed for the bribes.
\[
\sum_{v'}v'\left(\sum_{p,P,V' \mid v'=v(p,P,V')}x_{p,P,V'}\right)
\]

The first constraint ensures that the winning strategy of $e^*$ is at least $t$.
\[
\prod_p p^{\sum_{P}x_{p,P}}\ge t
\]

To make this constraint linear, we take the logarithm on both sides of the inequality and obtain the following.
\[
\sum_p \log p \left( \sum_{P} x_{p,P}\right)\ge \log t 
\]

Note that compared to the MILP for \cref{thm:numbervalues}, the objective function and the first constraint swap their roles.

Now, we have three additional sets of constraints, which are analogous to the ones for the MILP for \cref{thm:numbervalues}.
\[
\forall \ p,P \colon \sum_{V'}x_{p,P,V'} = x_{p,P}
\]
\[
\forall \ P, V' \colon \sum_p x_{p,P,V'} = n_{P,V'}
\]
\[
\forall \ p,P,V' \colon 0\le x_{p,P,V'}
\]

Note that the total number of variables and constraints is in $2^{O(p_\#)}\cdot n$ whereas the number of integer variables is in $2^{O(p_\#)}$, where $p_\#$ denotes the number of distinct probability values. By \cref{thm:MILP}, we can compute an optimal solution for the MILP in FPT-time with respect to the number~$p_\#$ of distinct probability values.
The correctness proof is analogous to the one of \cref{thm:numbervalues}.

\section{Conclusion}
In our work, we investigated the parameterized complexity of \probname, a natural tournament bribery problem. We extended work by \citet{mattei2015complexity} and established three main results:
\begin{itemize}
    \item \probname is W[1]-hard when parameterized by the number of players.
    \item \probname in FPT when parameterized by the number of distinct bribe values.
    \item \probname in FPT when parameterized by the number of distinct probability values.
\end{itemize}
To obtain our results, we established W[1]-hardness of \PK when parameterized by the number of items in the knapsack and developed a new way of designing MILPs that have optimal solutions where \emph{all} variables are set to integer values. Furthermore, we showed that our W[1]-hardness for \probname implies W[1]-hardness for \probnamecup for the same parameterization.

There are several natural directions for future research. It remains open whether \probname is NP-hard when the probabilities are encoded in unary and bribe values are encoded in binary (this question has already been raised by \citet{mattei2015complexity}). In fact, it is open whether \PK is NP-hard when the item values are encoded in unary and item sizes are encoded in binary. 

Furthermore, note that our FPT-algorithms have double-exponential running times in the parameter. We leave the question of whether this can be improved open. Moreover, exploring whether our MILP formulations for \probname can be extended to \probnamecup would be interesting. The main difficulty is that we heavily exploit in our MILP formulations that the ordering of the matches (that is, the seeding) is irrelevant in \probname, which is not the case in \probnamecup.

\bibliographystyle{abbrvnat}
\bibliography{bibliography}

\end{document}